\documentclass[final,3p,times,12pt]{elsarticle}

\usepackage{amssymb}
\usepackage{amsthm}

\usepackage{algorithm}
\usepackage[noend]{algpseudocode}
\usepackage{amsmath}
\usepackage{amsfonts}
\usepackage{url}

\usepackage{color}
\definecolor{hotcolor}{rgb}{1,0,0}

\newtheorem{mydef}{Definition}
\newtheorem{mylem}{Lemma}
\newtheorem{myprop}{Proposition}
\newtheorem{mythm}{Theorem}

\newtheorem{myrem}{Remark}
\newtheorem{myass}{Assumption}

\usepackage[active]{srcltx}

\usepackage{hyperref}  		
\hypersetup{colorlinks,linktocpage}

\journal{Journal of Computational Physics}

\begin{document}

\begin{frontmatter}



\title{\textbf{Numerical Method of Characteristics \\ for One--Dimensional Blood Flow}}



\author[label1]{Sebastian Acosta\corref{cor1}}
\ead{sacosta@bcm.edu}

\author[label2]{Charles Puelz}
\ead{cpuelz@rice.edu}

\author[label2]{B\'eatrice Rivi\`ere}
\ead{riviere@rice.edu}

\author[label1]{Daniel J. Penny}
\ead{djpenny@texaschildrens.org}

\author[label1]{Craig G. Rusin}
\ead{crusin@bcm.edu}

\address[label1]{Department of Pediatrics -- Cardiology, Baylor College of Medicine, Texas Children's Hospital, Texas}

\address[label2]{Department of Computational and Applied Mathematics, Rice University, Texas}

\cortext[cor1]{Corresponding author}

\begin{abstract}
Mathematical modeling at the level of the full cardiovascular system requires the numerical approximation of solutions to a one-dimensional nonlinear hyperbolic system describing flow in a single vessel. This model is often simulated by computationally intensive methods like finite elements and discontinuous Galerkin, while some recent applications require more efficient approaches (e.g. for real-time clinical decision support, phenomena occurring over multiple cardiac cycles, iterative solutions to optimization/inverse problems, and uncertainty quantification). Further, the high speed of pressure waves in blood vessels greatly restricts the time step needed for stability in explicit schemes. We address both cost and stability by presenting an efficient and unconditionally stable method for approximating solutions to diagonal nonlinear hyperbolic systems. Theoretical analysis of the algorithm is given along with a comparison of our method to a discontinuous Galerkin implementation. Lastly, we demonstrate the utility of the proposed method by implementing it on small and large arterial networks of vessels whose elastic and geometrical parameters are physiologically relevant.
\end{abstract}





\end{frontmatter}




\section{Introduction} \label{Section:Intro}

Recent research on hemodynamic models utilizes a set of equations describing blood flow in a single vessel.  In this model, the variables of interest are the vessel cross-sectional area $A = A(x,t)$ and the average blood velocity in the axial direction $x$ given as $u = u(x,t)$. Conservation of mass and balance of momentum respectively result in the following system of equations:
\begin{equation}
\label{eq:model}
\begin{aligned}
\frac{\partial A}{\partial t} + \frac{\partial (Au)}{\partial x} &= 0 \\
\frac{\partial u}{\partial t} + \frac{\partial}{\partial x}\left(\frac{u^2}{2} + \frac{p}{\rho}\right) &= - 8 \pi  \nu \frac{u}{A}.
\end{aligned}
\end{equation}
We call (\ref{eq:model}) the $(A,u)$--system, where $p=p(A)$ is the fluid pressure defined below in (\ref{eq:state}).  Here $\rho$ is the density of blood and $\nu$ is its kinematic viscosity. The assumptions of this model include the following: blood is an incompressible, viscous fluid flowing in a straight cylinder with compliant walls, and the characteristic length of the vessel (along the axial direction) is much larger than the characteristic radius.  A further assumption involves the functional form of the  velocity profile: for the left hand side of the momentum balance equation, a flat profile is assumed, whereas a parabolic profile is specified for the viscous term on right side. Other types of profiles for the viscous term may be used, see for example \cite{MWL09}. We follow the assumptions of \cite{FLQ03,Sh-Fr-Pei-Par-2003,MN08,SFPF03}.  This typical simplification, although inconsistent, is important since one may explicity compute the Riemann invariants of the system, and the viscous term remains finite.  For further details and a discussion of the related $(A,Q:=Au)$ model, see for example the works of Canic--Kim \cite{CK03}, Formaggia \textit{et al}. \cite{FGNQ01, FNQ02} and Sherwin \textit{et al}. \cite{SFPF03}. We choose to work with the $(A,u)$ system since our discontinuous Galerkin formulation is based on the work of Sherwin \textit{et al}.

To close the system, the functional relationship for the pressure $p$ is provided by the state equation
\begin{equation}
\label{eq:state}
p = p_{\rm ext} + \beta \big(A^{1/2} - A_{0}^{1/2}\big)
\end{equation}
where $p_{\rm ext}$ is the external pressure and $A_0$ is the vessel cross-sectional area for vanishing transmural pressure difference. The coefficient $\beta$ depends on the thickness, Young's modulus, Poisson's ratio and the unperturbed radius of the vessel \cite{SFPF03,Sh-Fr-Pei-Par-2003,MN08}. The above equation of state, which neglects viscoelasticity, renders the hyperbolicity of the system (\ref{eq:model})--(\ref{eq:state}). Since the method of characteristics is heavily reliant on hyperbolicity, we are not able to deal with viscoelastic effects in this paper. For detailed studies on the modeling and effects of viscoelasticity, see \cite{FLQ03,BGR08,AKK11,WFL14}.

These equations appear in recent literature for simulations of blood flow in a network of connected vessels, where system (\ref{eq:model}) models flow in each vessel and appropriate transmission conditions between vessels are specified.  As an example, vessel-network models of the full cardiovascular system provide important insight into different clinical and physiological questions.  Clinicians and engineers interested in the fluid dynamics around the heart may couple a 2d or 3d model of fluid flow close to the heart to a 1d network model of the arterial tree (and perhaps the venous tree).  This modeling approach has several benefits: first, the high fidelity 1d model of the arterial tree replaces overly simplistic lumped parameter models. Second, one may interrogate the 1d arterial tree model to better understand fluid flow in the peripheral circulation and the reflection of pressure waves.  Lastly, simulations of variants of the 1d model align well with experimental data from single tube, arterial model, and in vivo studies (see e.g. \cite{ANN09,OPK00,BGR08,AKK11,MAP07}). For some examples of 1d models derived from system (\ref{eq:model}) or the $(A,Q)$ system coupled to higher dimensional models see \cite{BFU07, FNQV99, BF13}.  Other clinical applications include stent flow simulations, models of fetus and neonate circulations, and surgical planning \cite{FNQ02, MWL09, Myn11,SFPF03}.  This collection of references, although not comprehensive, is meant to emphasize the versatility of (\ref{eq:model}).   
                  
Finite element, finite volume, discontinuous Galerkin, and other methods arising from weak formulations are successfully used for the spatial discretization of the $(A,u)$ or $(A,Q)$ systems \cite{DL13,FLQ03,SFPF03}. Although these methods maintain attractive mathematical properties, they are computationally intensive, and this complexity is magnified in simulations of vessel networks. For instance, the speed of pressure waves in blood vessels dictates the time step required for stability in explicit schemes. Unfortunately, for physiologically relevant choices of parameters, this speed may be much larger (at least one order of magnitude) than the velocity of blood flow. Moreover, this wave speed displays increasing variability as the arterial tree branches out \cite{Myn11}. This implies that the inclusion of smaller arteries in the model (to obtain more realistic and accurate simulations) may result in a more stringent stability condition. For a side--by--side comparision of several methods, see the recent paper by Wang \textit{et al}. \cite{WFL14}. These authors compare methods for simulating (\ref{eq:model}) based on several metrics, including running time for one cardiac cycle. 

In some instances, a more expensive discretization from a weak formulation is appropriate. But for our applications, we envision a 1d vessel--network model as a component in clinical decision support systems requiring simulation of multiple cardiac cycles. Fast iterative or repeated simulations are also needed for uncertainty quantification or to solve inverse problems via optimization \cite{Lassila-2013}. In these cases, close to real-time simulation is essential, and as such, the method for approximating solutions to (\ref{eq:model}) must be efficient and unconditionally stable. 
Fortunately, system (\ref{eq:model})--(\ref{eq:state}) has explicitly defined characteristic variables, and under the assumption of strict hyperbolicity, we may apply a numerical method of characteristics (NMC) for solving these equations. The method we propose is explicit in time (which makes it computational efficient) and unconditionally stable.

Many methods for numerically solving differential equations based on the characteristics have been proposed in the past. Some address the transport of a certain solvent or convection--dominated diffusion equations \cite{Dou-Rus-1982,Has-Liv-Ber-1983,Ew-Rus-Whe-1984,Rus-1985,Kri-Hay-Rus-1989,Suli-Ware-1991,Mor-Pri-Sul-1988}. Other works deal with approximations for Navier--Stokes equations in the absence of fluid--structure interaction \cite{Pir-1982,Sul-1988,Ach-Gue-2000} where the convective--derivative of the fluid velocity is treated with the method of characteristics. The concept behind the numerical method of characteristics also constitutes a main ingredient in the CIP method developed in \cite{Tak-Yab-1987,Ida-Yab-1995,Tan-Nak-Yab-2000,Yoon-Yabe-1999}. Furthermore, variants of this method have very recently been applied in the hemodynamics context \cite{MG08,KVS14,BH12}. Unfortunately, these latter publications do not rigorously address stability and convergence. Wang and Parker \cite{Wang-Parker-2004} also propose a method of characteristics for simulating circulation in the arterial network. However, in contrast to our work, they consider a fully linearized approach where the nonlinearities arising from convection and the pressure dependent wave speed are neglected. For a quantification of these nonlinear effects, see \cite{Sh-Fr-Pei-Par-2003,Myn-Dav-Mal-Pen-Smo-2012} and references therein. 

This manuscript details the application of the NMC to fully nonlinear blood flow (Sections \ref{Section:Char}-\ref{Section:Algorithm}) and develops the standard numerical analysis including stability and convergence (Section \ref{Section:Analysis}). Our analysis is supported with numerical experiments to confirm the proven rate of convergence and to compare the NMC with a discontinous Galerkin (dG) discretization of (\ref{eq:model}). We conclude with an application of the NMC method to an arterial network of vessels (Section \ref{Section:NumExp}).


\section{Characteristics for one-dimensional blood flow} \label{Section:Char}

In this section, we recapitulate some useful mathematical properties of (\ref{eq:model}).  First, let us consider a general system of the form: 
\begin{equation}
\label{eq:gen1}
\frac{\partial {\bf U}}{\partial t} + \frac{\partial {\bf F}({\bf U})}{\partial x} = {\bf S}({\bf U})
\end{equation}
where ${\bf U} \in \mathbb{R}^2$ \Big(${\bf U} = (u_1,u_2)^T$\Big). This system may be written in a quasilinear form, namely
\begin{equation}
\label{eq:gen2}
\frac{\partial {\bf U}}{\partial t} + \nabla_{\bf U}{\bf F}\frac{\partial {\bf U}}{\partial x} = {\bf S}({\bf U})
\end{equation}
where $\nabla_{{\bf U}}{\bf F}$ is the $2 \times 2$ Jacobian matrix of ${\bf F}$ and the source function may change to include some terms from differentiating ${\bf F}$.  As we shall see, (\ref{eq:model}) may be expressed in this form.  Let the left eigenvectors of $\nabla_{\bf U}{\bf F}$ be given as $\{ {\bf l}_1({\bf U}), {\bf l}_2 ({\bf U}) \}$ with eigenvalues $\{\lambda_1 ({\bf U}), \lambda_2 ({\bf U}) \}$ (we will henceforth drop the notation indicating their dependence on ${\bf U}$).  The system (\ref{eq:gen1}) is \textit{strictly hyperbolic} provided the Jacobian matrix has real distinct eigenvalues.

The general idea for the method of characteristics is to transform system (\ref{eq:gen1}) by diagonalizing the principal part of the differential equation in the hope that one finds functions remaining constant along particular curves. With this in mind, consider $Z_i:\mathbb{R}^2 \rightarrow \mathbb{R}$ whose gradient $\nabla_{\bf U}Z_i$ is parallel to ${\bf l}_i$; these are called \textit{Riemann--invariants} (see e.g. \cite[p. 637]{Evans2010}).  Now, define functions $V_1$ and $V_2$ from $Z_1$ and $Z_2$ like 
\begin{align}
\label{eq:char1}
V_1(x,t)&=Z_1({\bf U}(x,t)) + k_{1}(x,t), \\
\label{eq:char2}
V_2(x,t)&=Z_2({\bf U}(x,t)) + k_{2}(x,t),
\end{align}
where $k_{i}$ are arbitrary constants of integration, that is, $\nabla_{\bf U} k_{i} = 0$. We refer to $V_1$ and $V_2$ as the \textit{characteristics variables} of system (\ref{eq:gen1}).  From the chain rule combined with (\ref{eq:gen2}), $V_1$ and $V_2$ satisfy
\begin{align}
\label{eq:char3}
\frac{\partial V_1}{\partial t} + \lambda_1 \frac{\partial V_1}{\partial x}  &= R_{1} :=  \nabla_{\bf U}Z_1^T {\bf S}({\bf U}) + \frac{\partial k_1}{\partial t} + \lambda_1 \frac{\partial k_1}{\partial x} \\
\label{eq:char4}
\frac{\partial V_2}{\partial t} + \lambda_2 \frac{\partial V_2}{\partial x}  &= R_{2} :=  \nabla_{\bf U}Z_2^T {\bf S}({\bf U}) + \frac{\partial k_2}{\partial t} + \lambda_2 \frac{\partial k_2}{\partial x}.
\end{align}
The next statement is important for our method. It is easy to see that the following holds.
\begin{myprop}
\label{curves}
The function $V_i(x,t) - \int_0^t R_{i}(x,s) ds$ is constant along the curve $(\gamma_i (s), s)$ satisfying
\[
\frac{d{\gamma}_i}{ds} = \lambda_i(\gamma_i (s), s).
\] 
\end{myprop}
\noindent

We derive the characteristic variables for system (\ref{eq:model}) by following equations (\ref{eq:char1}) -- (\ref{eq:char4}) with Proposition \ref{curves}. Assuming constant $\beta$, we rewrite the system with the Jacobian of ${\bf F}$ as follows,
\[
\underbrace{\frac{\partial}{\partial t}\begin{bmatrix}
A \\
u
\end{bmatrix}}_{\partial {\bf U}/ \partial t} + 
\underbrace{\begin{bmatrix}
u & A \\ 
c^2 / A & u 
\end{bmatrix}}_{\nabla_{\bf U} {\bf F}}
\underbrace{\frac{\partial}{\partial x}
\begin{bmatrix}
A \\
u
\end{bmatrix}}_{\partial {\bf U}/ \partial x} = \underbrace{
\begin{bmatrix} 0 \\ 
- 8 \pi \nu \frac{u}{A}  + 4 c_{0} \frac{d c_0}{dx}
 \end{bmatrix}}_{{\bf S}(\bf U)},
\]
where the perturbed and unperturbed wave speeds are given by
\begin{align}
c = c(A) = \left(\frac{\beta \sqrt{A}}{2 \rho}\right)^{1/2} \quad \text{and} \quad  c_{0} = c(A_{0}). \label{eqn:speed}
\end{align}
The left eigenvectors and eigenvalues for $\nabla_{\bf U}{\bf F}$ are
\begin{align}
&\lambda_1 = u + c, \qquad {\bf l}_1 = \begin{bmatrix} c/A \\ 1 \end{bmatrix}, \label{eq:Eigen01} \\
&\lambda_2 = u - c, \qquad {\bf l}_2 = \begin{bmatrix} -c/A \\ 1 \end{bmatrix}. \label{eq:Eigen02}
\end{align}
If we set $\nabla_{\bf U}Z_1 = {\bf l}_1$ and $\nabla_{\bf U}Z_2  = {\bf l}_2$, then with ${\bf U} = (A,u)^{T}$ we have
\begin{align*}
&\frac{\partial Z_1}{\partial A} = \frac{c}{A}, \qquad \frac{\partial Z_1 }{\partial u} = 1, \\
&\frac{\partial Z_2 }{\partial A} = -\frac{c}{A}, \qquad \frac{\partial Z_2}{\partial u} = 1.
\end{align*}
For convenience we choose $k_{1} = - 4 c_{0}$ and $k_{2} = 4 c_{0}$. Integrating, we obtain:
\begin{align}
\label{eq:v1v2_1}
V_1(x,t) &= u(x,t) + 4 \left( c(A(x,t)) - c_{0}(x) \right), \\
\label{eq:v1v2_2}
V_2(x,t) &= u(x,t) - 4 \left(  c(A(x,t)) - c_{0}(x) \right) ,
\end{align}
where these variables satisfy the system
\begin{equation}
\label{eq:1dchar}
\begin{aligned}
\frac{\partial V_1}{\partial t} + (u+c) \frac{\partial V_1}{\partial x}  &= R_{1}  = - 8 \pi \nu \frac{u}{A} - 4 (u + c - c_{0}) \frac{d c_{0}}{d x},  \\
\frac{\partial V_2}{\partial t} + (u-c) \frac{\partial V_2}{\partial x}  &= R_{2} = - 8 \pi \nu \frac{u}{A} + 4 (u - c + c_{0}) \frac{d c_{0}}{d x}.
\end{aligned} 
\end{equation}
One may recover the cross-sectional area (and hence the pressure or wave speed) and velocity from the characteristic variables, and vice versa. Specifically, 
\begin{equation}
u = \frac{V_{1} + V_{2}}{2} \quad \text{and} \quad c - c_{0} = \frac{V_{1} - V_{2}}{8}. \label{eq:Ch_vs_Phys}
\end{equation}

The above derivation reveals that the characteristic variables propagate at speeds $u \pm c$, where $u$ is the velocity of blood. For physiologically relevant parameter values, $c \gg |u|$.  In particular, this relationship between $u$ and $c$ implies that $\lambda_1>0$ and $\lambda_2<0$, that is, the characteristic variables propagate in opposite directions.

Most explicit time discretizations require a CFL--type restriction on the timestep determined by $c$ despite the fact that the speed of blood $u$ is much smaller. To avoid this strong restriction, we propose a method that is stable regardless of the chosen timestep.


\section{Algorithm} \label{Section:Algorithm}

For the presentation of the algorithm, let us focus on the following initial value problem,
\begin{eqnarray}
\label{eq:orig1}
\frac{\partial V_1}{\partial t} + \lambda_1(V_{1},V_{2},x,t) \frac{\partial V_1}{\partial x} &=& R_{1}(V_1,V_2,x,t) \\
\frac{\partial V_2}{\partial t} + \lambda_2(V_{1},V_{2},x,t) \frac{\partial V_2}{\partial x} &=& R_{2}(V_1,V_2,x,t) \\
V_1(x,0) &=& V_{1}^0(x) \\
\label{eq:orig2}
V_2(x,0) &=& V_{2}^0(x)
\end{eqnarray}
defined on intervals $x \in [a,b]$ and $t \in [0,T]$, and augmented by periodic boundary conditions of the form
\begin{equation*}
V_i(a,t) = V_i(b,t) \quad i = 1,2.
\end{equation*}

Now we introduce some notation. We use the following supremum norms in our analysis:
\begin{equation}
\|q\|:= \sup_{x \in [a,b]} |q(x)|  \quad \text{and}  \quad
\|p\|_T := \sup_{x\in[a,b], \,\, t \in [0,T]} |p(x,t)|.
\end{equation}
Let dashes denote derivatives in space and dots denote derivatives in time, i.e.  $p' := \partial p / \partial x$ and $\dot{p} := \partial p / \partial t$.  For the spatial discretization, let $G_h:= \Big\{x_j = a+ j(b-a)/M, \text{ }j=0, \ldots M \Big\}$,
i.e. the collection of uniformly spaced points between $a$ and $b$ with spacing $h := (b-a)/M$.  Define $\mathcal{C}[a,b]$ to be the space of continuous functions on $[a,b]$, and $\mathcal{C}_h[a,b]$ to be the subset of continuous functions that are linear when restricted to each interval $[x_j,x_{j+1}]$ for $j = 0, \ldots M-1$. For the temporal discretization, given a positive integer $N$, define the timestep $\Delta t: = T/N$ and $t_n := n\Delta t$.

In what follows, $V_i$ refers to the exact solution whereas $W_i$ refers to the approximate solution.
The numerical method of characteristics for solving (\ref{eq:orig1}) -- (\ref{eq:orig2}) is based on the following idea: to obtain an approximation $W_i$ to $V_i$ given information on the grid $G_h$, follow the movement of the points in $G_h$ along the characteristic curves \textit{back} in time, and then assign values at the \textit{current} time via spatial interpolation of the solution.  More explicitly, from Proposition \ref{curves} with $\gamma_i(t+\Delta t) = x \in G_h$ one has
\begin{eqnarray}
\label{eq:intchar3}
V_i(x, t+\Delta t) = V_i(\gamma_i(t), t) + \int_{t}^{t+\Delta t} R_{i}(\gamma_{i}(s),s) ds.
\end{eqnarray} 
With this in mind, we have the following set of definitions.  For each $x \in [a,b]$ define the characteristic curve $\gamma_i(x,t_{n+1};t):[t_n,t_{n+1}] \rightarrow \mathbb{R}$ passing through point $x$ at time $t_{n+1}$ as the solution to the following \textit{final} value problem:
\begin{equation}
\label{eq:intchar2}
\begin{aligned}
&\frac{d\gamma_i(x,t_{n+1};t)}{dt} = \lambda_i\big(\gamma_i(x,t_{n+1};t),t\big) \\
&\gamma_i(x,t_{n+1};t_{n+1}) = x.
\end{aligned}
\end{equation}

\begin{mydef}
\label{def:curves}
Let $n=1,2,...,N$. For $x \in [a,b]$, let $\tilde{g}_i^n(x)$ ($i = 1,2$) be an approximation to the quantity 
\begin{equation}
g^n_i(x)= x - \mathcal{I}_i^n(x) :=x - \int_{t_n}^{t_{n+1}}\lambda_i\big(\gamma_i(x,t_{n+1};t),t\big)dt
\end{equation}
in the sense that
\begin{equation}
\tilde{g}_i^n(x):= x - \tilde{\mathcal{Q}}^n_i(x)
\end{equation} 
where $\tilde{\mathcal{Q}}_i^n$ is a ``pseudo--quadrature rule'' for the integral $\mathcal{I}_i^n$ computed with the approximate solution $W_i$. Define $\mathcal{Q}_i^n$ to be this same pseudo--quadrature rule computed with the exact solution $V_i$. As we will see below, the rule we define is equivalent to a linearization of the characteristic curve. An illustration of the definition of $g^{n}(x)$ and and $\tilde{g}^{n}(x)$ is displayed in Figure \ref{fig:curves}. Note that $g^n_i(x)$ and $\tilde{g}_i^n(x)$ may not lie in the interval $[a,b]$, but its definition can be easily adjusted to handle the periodic boundary condition.
\end{mydef}

\begin{figure}[H]
\begin{center}
\includegraphics[width = 0.9 \textwidth]{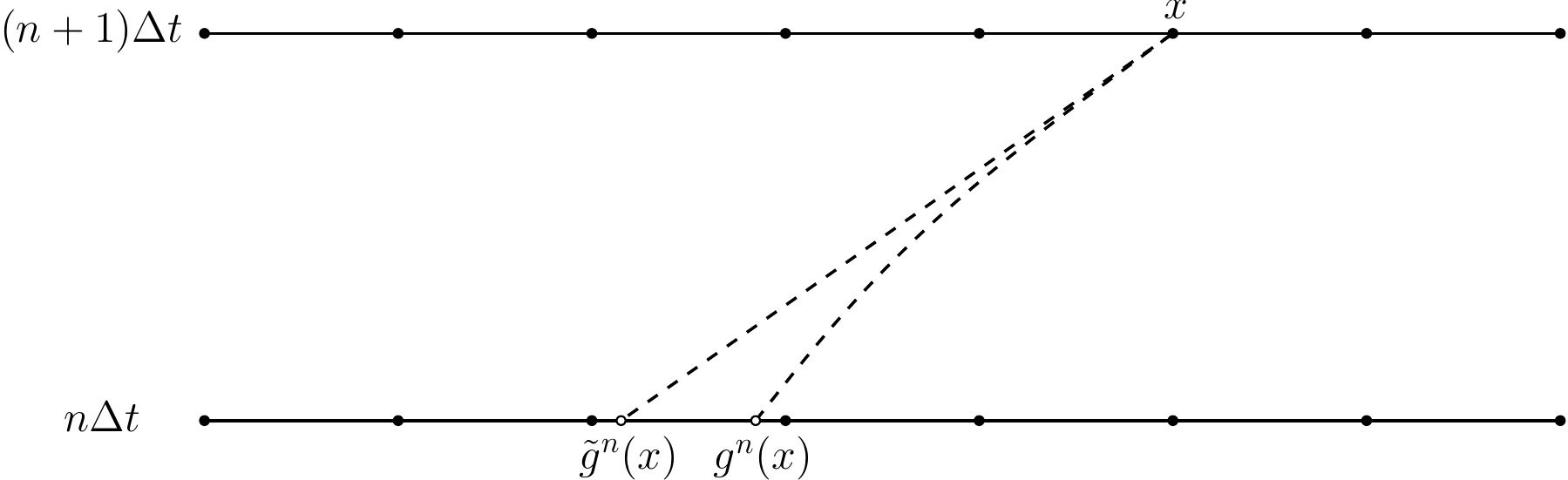}
\end{center}
\caption{The characteristic curve and its approximation. The \textit{head} of the characteristic curve is the grid point $x$, and its \textit{foot} is denoted by $g^{n}(x)$. The approximate \textit{foot}, denoted by $\tilde{g}^{n}(x)$, is obtained by a linearization of the characteristic curve given in Definition \ref{def:curves}.}
\label{fig:curves}
\end{figure}

\noindent
Take $x \in [a,b]$ and consider the characteristic curve within the time interval $[t_n,t_{n+1}]$ on which $x$ lies at time $t_{n+1}$, i.e. $\gamma_i(x,t_{n+1},t)$. To declutter notation, define $V_{i}^{n}(x) = V_{i}(x,t_{n})$ for all $n$. By Definition \ref{def:curves} and (\ref{eq:intchar2}) we have $g_i^n(x) = \gamma_i(x,t_{n+1};t_{n})$. In turn, for the solution $V_i$ one has
\begin{align*}
V_i^{n+1}(x) = V_i(\gamma_i(x,t_{n+1};t_n),t_n) + \mathcal{J}_{i}^{n}(x)  := V_i^n(g_i^n(x)) + \int_{t_{n}}^{t_{n+1}} R_{i} (\gamma_{i}(x,t_{n+1};t),t) dt.
\end{align*}
We have shown the following lemma which is nothing more than rewriting (\ref{eq:intchar3}) in more compact notation.
\begin{mylem}
\label{lem:propa}
The solutions $V_i$ to (\ref{eq:orig1}) -- (\ref{eq:orig2}) satisfy
\begin{equation}
V_i^{n+1}(x) = V_i^n(g^n(x)) + \mathcal{J}_{i}^{n}(x)  \quad \text{ for all } x \in [a,b] \text{ and } n = 1 \ldots N.
\end{equation} 
\end{mylem}

To define the quadrature rule $\mathcal{Q}_i^n$ ( and hence $\tilde{\mathcal{Q}}_i^n$ ), we recall that $\lambda_i$ is a function of the characteristic variables $V_1$ and $V_2$.  For example, for the blood flow system (\ref{eq:model})--(\ref{eq:state}), combining (\ref{eq:Eigen01})--(\ref{eq:Eigen02}) and (\ref{eq:v1v2_1})--(\ref{eq:v1v2_2}), one has,
\begin{align}
\label{eq:eigs1}
\lambda_1(x,t) = \frac{5}{8}V_1(x,t) + \frac{3}{8}V_2(x,t) + c_{0}(x) \quad \text{and} \quad
\lambda_2(x,t) = \frac{3}{8}V_1(x,t) + \frac{5}{8}V_2(x,t) - c_{0}(x),
\end{align}
so in accordance with our previous notation, we can write $\lambda_i(x,t_n) = \lambda_i(V_1^n(x),V_2^n(x),x,t_{n})$.  In turn, we would like to approximate the integral by the simplest ``rectangle rule'', i.e.
\begin{equation}
\mathcal{I}_i^n(x) \approx \Delta t \lambda_i(\gamma_i(x,t_{n+1};t_n),t_n) = \Delta t \lambda_i(V_1^n(g_i^n(x)),V_2^n(g_i^n(x)),g^{n}_{i}(x),t_{n}).
\end{equation}
Let us define $\mathcal{Q}_{i,R}$ and $\tilde{\mathcal{Q}}_{i,R}$ via the rectangle rule approximation:
\begin{align}
\mathcal{Q}_{i,R}^n(x)&:= \Delta t \lambda_i(V_1^n(g_i^n(x)),V_2^n(g_i^n(x)),g^{n}_{i}(x),t_{n}) \\
\tilde{\mathcal{Q}}_{i,R}^n(x)&:= \Delta t \lambda_i(W_1^n(\tilde{g}_i^n(x)),W_2^n(\tilde{g}_i^n(x)),\tilde{g}^{n}_{i}(x),t_{n}),
\end{align}
where $\tilde{\mathcal{Q}}_{i,R}^n$ is computed with the approximate solution $W_1^n$, $W_2^n$.  

\begin{myrem}
If we were to take our pseudo--quadrature rule to be $\mathcal{Q}_i^n = \mathcal{Q}^n_{i,R}$ and $\tilde{\mathcal{Q}}^n_i = \tilde{\mathcal{Q}}^n_{i,R}$, then the formula to determine $\tilde{g}_i^n(x)$ becomes nonlinear and hence \textit{implicit} in time, i.e.
\begin{equation*}
\tilde{g}_i^n(x) = x - \Delta t \lambda_i(W_1^n(\tilde{g}_i^n(x)),W_2^n(\tilde{g}_i^n(x)),\tilde{g}^{n}_{i}(x),t_{n}):=\mathcal{K}_i^n(\tilde{g}_i^n(x)).
\end{equation*}
For small enough $\Delta t$, $\mathcal{K}^n$ is a contraction.  If the rectangle rule scheme is employed, $\tilde{g}^n(x)$ may be computed as the limit of the sequence $y^{(k+1)} = \mathcal{K}^n(y^{(k)})$ with initial condition $y^{(0)}=x$.
\end{myrem}

To simplify the method and have an explicit time stepping procedure, we define the rule we implement from the rectangle rule by replacing both $g_i^n(x)$ and $\tilde{g}_i^n(x)$ with $x$ in both $\mathcal{Q}^n_{i,R}$ and $\tilde{\mathcal{Q}}^n_{i,R}$ respectively. 

Similarly, the source term $R_{i}$ may be a function of the characteristic variables $V_{1}$ and $V_{2}$ so that $R_{i}(x,t) = R_{i}(V_{1}(x,t),V_{2}(x,t),x,t)$. We approximate the exact integral $\mathcal{J}_{i}^{n}$ using a similar explicit quadrature rule denoted by $\tilde{\mathcal{R}}_{i}^n$. More precisely, we have the following definition:

\begin{mydef} \label{Def.PseudoQuad}
The pseudo--quadrature rules applied to the exact and approximate solutions are defined as follows:
\begin{align*}
&\mathcal{Q}_{i}^n(x) := \Delta t \lambda_i(V_1^n(x),V_2^n(x),x,t_{n}) \quad \text{and} \quad \mathcal{R}_{i}^n(x) := \Delta t R_{i}(V_1^n(\tilde{g}^{n}_{i}(x) ),V_2^n(\tilde{g}^{n}_{i}(x)),\tilde{g}^{n}_{i}(x),t_{n})
, \\
&\tilde{\mathcal{Q}}_{i}^n(x) := \Delta t \lambda_i(W_1^n(x),W_2^n(x),x,t_{n}) \quad \text{and} \quad
 \tilde{\mathcal{R}}_{i}^n(x) := \Delta t R_{i}(W_1^n(\tilde{g}^{n}_{i}(x)),W_2^n(\tilde{g}^{n}_{i}(x)),\tilde{g}^{n}_{i}(x),t_{n}).
\end{align*}
\end{mydef}

\noindent
The last missing piece is the specification of the spatial interpolation procedure.
\begin{mydef}
$\Pi_h:\mathcal{C}[a,b] \rightarrow \mathcal{C}_{h}[a,b]$ projects a continuous function $f$ into its piecewise linear interpolant $\Pi_h f$ at the points in $G_h$. 
\end{mydef}
\noindent
The algorithm follows below.
\noindent
\begin{algorithm}
\begin{algorithmic}[]
\State {\bf Input:} $V_1^0, V_2^0 \in \mathcal{C}[a,b]$.
\State Initialize $W_1^0 = \Pi_{h}[V_1^0]$ and $W_2^0 = \Pi_{h}[V_2^0]$.
\State {\tt for} $n = 1, 2, \ldots N$
\State \hspace{1cm}$\tilde{g}_{i}^{n-1}(x) = x - \tilde{\mathcal{Q}}_{i}^{n-1}(x)$ \hspace{1cm} $i=1,2$
\State \hspace{1cm}$W_{i}^n(x) = \Pi_{h} [W_{i}^{n-1}(\tilde{g}_{i}^{n-1}(x)) + \tilde{\mathcal{R}}_{i}^{n-1}(x) ]$ \hspace{1cm} $i=1,2$
\State {\tt end}
\end{algorithmic}
\caption{NMC algorithm for system (\ref{eq:orig1})-(\ref{eq:orig2})}
\label{alg:NMC}
\end{algorithm}
\noindent

\begin{myrem}
Higher order interpolation and quadrature is possible.  We work with piecewise linear interpolation for our analysis since the norm of $\Pi_h$ is uniformly bounded by $1$ for all $h$ which leads to stability. Also, the rule defined in Definition \ref{Def.PseudoQuad} allows our method to remain explicit in time.
\end{myrem}

\begin{myrem}
In practice, we compute the approximate solution $W_i$ at the points in $G_h$, but in the presentation of the algorithm above, the approximate solution is viewed equivalently as a piecewise linear function in $\mathcal{C}_h[a,b]$.  We use this presentation since we work with the continuous supremum norm for our analysis.    
\end{myrem}


\section{Numerical Analysis} \label{Section:Analysis}

Let $V_1(x,t^n)$, $V_2(x,t^n)$ and $W_1^n(x)$, $W_2^n(x)$ be the exact and approximate solutions to (\ref{eq:orig1}) -- (\ref{eq:orig2}) respectively.  We make the following assumptions:
\begin{myass}
\label{ass:V}
The exact solutions satisfy $V_i \in \mathcal{C}^2([0,T]\times[a,b])$.
\end{myass}

\begin{myass}
\label{ass:lam}
The eigenvalues $\lambda_i = \lambda_i(V_{1},V_{2},x,t)$ are continuously differentiable. Also, there are positive constants $\delta$ and $K=K(\delta)$ so that in the domain $\left(|V_1| + |V_2|\right)<\delta$ the source functions $R_{i}=R_{i}(V_1,V_2,x,t)$ are continuously differentiable and satisfy $|R_{i}(V_1,V_2)| \leq K \left( | V_1 | + | V_2 | \right)$. 
\end{myass}

\noindent
Note that Assumption \ref{ass:lam} regarding $\lambda_{i}$ holds for the blood flow system (\ref{eq:model})--(\ref{eq:state}) because the eigenvalues $\lambda_i$ are affine functions of the characteristic variables, as verified in (\ref{eq:eigs1}). Assumption \ref{ass:lam} concerning $R_{i}$ is satisfied if the cross-sectional area $A(x,t)$ is bounded away from zero uniformly in space and time, which is guaranteed when $\| V_{1} \| + \| V_{2} \|$ is sufficiently small. In turn, we need our numerical solution $(W_{1}^n,W_{2}^n)$ to satisfy the same property up to some finite time $T$ so that $R_{i}$ remains sufficiently smooth along the trajectory of the numerical solution. This is ensured by the following proposition.

\begin{myprop}[Stability] \label{Prop.Stability}
Under Assumption \ref{ass:lam}, if $\left( \| W_{1}^{0} \| + \| W_{2}^{0} \| \right) < \delta e^{-2 K T}$, then 
\begin{equation*}
\| W_{1}^{n} \| + \| W_{2}^{n} \| \leq  e^{2 K T} \left(  \| W_{1}^{0} \| + \| W_{2}^{0} \| \right) < \delta, \quad n=1,...,N.
\end{equation*}
\end{myprop}

\begin{proof}
We rely on the fact that for piecewise linear interpolation we have $\|\Pi_h\| = 1$. We proceed by induction. Assume that
\begin{align*}
\| W_{1}^{m} \| + \| W_{2}^{m} \|  < \delta, \quad \text{for all $m=0,...,n-1$,}
\end{align*}
and consider the following inequality,
\begin{align*}
& \| W_{1}^{n} \|  \leq  \| W_{1}^{n-1}(\tilde{g}_{1}^{n-1}) + \Delta t R_{1}(W_{1}^{n-1}(\tilde{g}_{1}^{n-1}),W_{2}^{n-1}(\tilde{g}_{1}^{n-1}) ) \| \leq \| W_{1}^{n-1} \| + \Delta t K \left( \|W_{1}^{n-1}\| + \| W_{2}^{n-1} \| \right), \\
& \| W_{2}^{n} \|  \leq  \| W_{2}^{n-1}(\tilde{g}_{2}^{n-1}) + \Delta t R_{2}(W_{1}^{n-1}(\tilde{g}_{2}^{n-1}),W_{2}^{n-1}(\tilde{g}_{2}^{n-1}) ) \| \leq \| W_{2}^{n-1} \| + \Delta t K \left( \|W_{1}^{n-1}\| + \| W_{2}^{n-1} \| \right).
\end{align*}
Therefore, 
\begin{align*}
\| W_{1}^{n} \| + \| W_{2}^{n} \| \leq 
\left(1 + 2 K \Delta t \right) \left( \|W_{1}^{n-1} \| + \| W_{2}^{n-1} \| \right) \leq e^{2 K T} \left( \|W_{1}^{0} \| + \| W_{2}^{0} \| \right) < \delta,
\end{align*}
where the second inequality follows by recursion and the strong inductive hypothesis. The last inequality follows from the assumption on the initial condition. This concludes the proof.
\end{proof}

\begin{myrem}
We wish to comment on the physical meaning of Assumption \ref{ass:lam}. When the characteristics variables $(V_{1},V_{2})$ are sufficiently small, the cross-sectional area $A$ is positive and the velocity $u$ remains bounded. This prevents the solution from going into the vacuum state corresponding to $A = 0$, i.e. vessel collapse. Further, a sufficiently small constant $\delta$ in Assumption \ref{ass:lam} can be estimated from the unperturbed wave speed $c_{0}$ as $\delta < 8  \inf_{x} c_{0}(x)$.
\end{myrem}

A convergence result for the algorithm follows below.
\begin{mythm}[Convergence] \label{Thm.Convergence}
Fix $T>0$ and $\Delta t = T/N$ for $N \in \mathbb{N}$. Under Assumptions \ref{ass:V} and \ref{ass:lam}, and the hypothesis from Proposition \ref{Prop.Stability} on the initial condition $(W_{1}^{0},W_{2}^{0})$, the following convergence bound holds:
\begin{equation*}
  \|W_1^n -V_1^n\| + \|W_2^n - V_2^n\| \leq T \exp(C T) \big[\mathcal{O}(h^2 /\Delta t) + \mathcal{O}(\Delta t) \big] \qquad \text{for all $n=1,2,...,N$,}
\end{equation*}
for some positive constant $C = C(V_{1},V_{2})$.
\end{mythm}
\begin{proof}
We first bound $\|W_1^n - V_1^n\|$.  One has $\| W_1^n - V_1^n \| \leq \| W_1^n -  \Pi_h V_1^{n}\| + \| \Pi_h V_1^{n} - V_1^n \|$.  We apply Lemma \ref{lem:propa} to plug in $V_1^n = V_1^{n-1}(g_1^{n-1}) + \mathcal{J}_{1}^{n-1}$, use $\| \Pi_{h} \| = 1$, and then bound the first term as follows.
\begin{align*}
\| W_1^n - \Pi_h V_1^n\| &\leq \| W_1^{n-1}(\tilde{g}_1^{n-1}) - V_1^{n-1}( \tilde{g}_1^{n-1}) \| + \| V_1^{n-1}(\tilde{g}_1^{n-1}) - V_1^{n-1}(g_1^{n-1}) \| \\
& \quad + \| \tilde{\mathcal{R}}_{1}^{n-1} - \mathcal{R}_{1}^{n-1} \| + \| \mathcal{R}_{1}^{n-1} - \mathcal{J}_{1}^{n-1} \|
\\
&\leq \| W_1^{n-1} - V_1^{n-1}\| + \|(V_1^{n-1})'\|\|\tilde{g}_1^{n-1} - g_1^{n-1}\| + \| \tilde{\mathcal{R}}_{1}^{n-1} - \mathcal{R}_{1}^{n-1} \| + \| \mathcal{R}_{1}^{n-1} - \mathcal{J}_{1}^{n-1} \|
\end{align*}
To bound $\| \tilde{g}_1^{n-1} - g_1^{n-1} \|$, note that for any $x$, we have 
\begin{align*}
|\tilde{g}_1^{n-1}(x) - g_1^{n-1}(x)| = |\mathcal{I}_1^{n-1}(x) - \tilde{\mathcal{Q}}_1^{n-1}(x)| 
\leq |\mathcal{I}_1^{n-1}(x) - \mathcal{Q}_{1,R}^{n-1}(x)| + | \mathcal{Q}_{1,R}^{n-1}(x) -  \tilde{\mathcal{Q}}^{n-1}_1(x)|.
\end{align*}
The first term is the quadrature error due to the rectangle rule and the second term may be bounded in the following way:
\begin{align*}
|&\mathcal{Q}_{1,R}^{n-1}(x) -  \tilde{\mathcal{Q}}^{n-1}_1(x)| \leq | \mathcal{Q}_{1,R}^{n-1}(x) -  {\mathcal{Q}}^{n-1}_1(x)| + | \mathcal{Q}^{n-1}_1(x) -  \tilde{\mathcal{Q}}^{n-1}_1(x)| \\
&= \Delta t \, |\lambda_1\big(V_1^{n-1}(g_1^{n-1}(x)), V_2^{n-1}(g_1^{n-1}(x))\big) - \lambda_1\big(V_1^{n-1}(x),V_2^{n-1}(x)\big)|\\
 &\quad + \Delta t  \, |\lambda_1(V_1^{n-1}(x),V_2^{n-1}(x)) - \lambda_1( W_1^{n-1}(x), W_2^{n-1}(x))| \\
&\leq \Delta t C_{\lambda} \Big\{ |V_1^{n}(g_1^{n-1}(x)) - V_1^{n-1}(x)|  + |V_2^{n}(g_1^{n-1}(x)) - V_2^{n-1}(x)| \\
 &\quad + |W_1^{n-1}(x) - V_1^{n-1}(x)| + |W_2^{n-1}(x) - V_2^{n-1}(x)|  \Big\} \\
&\leq \Delta t C_{\lambda} \Big\{ \|(V_1^{n-1})'\||g_1^{n-1}(x) - x| + \|(V_2^{n-1})'\||g_1^{n-1}(x) - x| \\
 &\quad + |W_1^{n-1}(x) - V_1^{n-1}(x)| + |W_2^{n-1}(x) - V_2^{n-1}(x)|  \Big\} \\
&\leq \Delta t^2 \|\lambda_1\|_T C_{\lambda} \Big\{ \|(V_1^{n-1})'\| + \|(V_2^{n-1})'\|\Big\}+ \Delta t C_{\lambda}\Big\{|W_1^{n-1}(x) - V_1^{n-1}(x)| +|W_2^{n-1}(x) - V_2^{n-1}(x)| \Big\}.
\end{align*}
With this bound, one has
\begin{align*}
\|\tilde{g}_1^{n-1} - g_1^{n-1}\| &\leq \| \mathcal{I}_1^{n-1} - \mathcal{Q}_{1,R}^{n-1} \| +  \Delta t C_{\lambda} \Big\{\| W_1^{n-1} - V_1^{n-1} \| + \| W_2^{n-1} - V_2^{n-1} \| \Big\} \\
&+ \Delta t^2 \|\lambda_1\|_T C_{\lambda} \Big\{ \|(V_1^{n-1})'\| + \|(V_2^{n-1})'\|\Big\}.
\end{align*}
Now we proceed to bound the term $\| \tilde{\mathcal{R}}_{1}^{n-1} - \mathcal{R}_{1}^{n-1} \|$ as follows. From Assumption \ref{ass:lam}, we get 
\begin{align*}
\| \tilde{\mathcal{R}}_{1}^{n-1} - \mathcal{R}_{1}^{n-1} \| \leq \Delta t \, C_{R} \big\{ \| W_{1}^{n-1} - V_{1}^{n-1} \| + \| W_{2}^{n-1} - V_{2}^{n-1} \| \big\}
\end{align*}
where $C_R$ is a Lipschitz constant working for both $R_1$ and $R_2$. Similarly,
\begin{align*}
\| \mathcal{R}_{1}^{n-1} - \mathcal{J}_{1}^{n-1} \| &\leq \Delta t  C_{R} \big\{ \| V_{1}^{n-1}(\tilde{g}^{n-1}_{1}) - V_{1}^{n-1}(g^{n-1}_{1}) \| + \| V_{2}^{n-1}(\tilde{g}^{n-1}_{1}) - V_{2}^{n-1}(g^{n-1}_{1}) \| \big\}) + \hat{C} \Delta t^2 \\
&\leq \Delta t C_{R} \big\{ \| (V_{1}^{n-1})'\| + \| (V_{2}^{n-1})'\|  \big\} \|\tilde{g}_1^{n-1} - g_1^{n-1}\| + \hat{C} \Delta t^2,
\end{align*}
where the last term is obtained by approximating the integral $\mathcal{J}_{1}^{n-1}$ by the rectangle rule and employing the differentiability of $R_1$ and of the exact solution $V_{i}$.

With Assumption \ref{ass:V}, we choose a constant $\tilde{C}$ that simultaneously bounds the terms involving $C_{R}$, $C_{\lambda}$, $\|\lambda_i\|_T$ and the norm of the first derivative of $V_i^{n}$ for $i = 1, 2$ and $n = 1 \ldots N$. Then we have,
\begin{align*}
\| W_1^n - V_1^n \| &\leq (1+ \Delta t \tilde{C})\| W_1^{n-1} - V_1^{n-1} \| + \Delta t\tilde{C} \|W_2^{n-1} - V_2^{n-1}\| + \|\Pi_{h} V_1^{n} - V_1^{n} \| \\
& \qquad + \tilde{C} (1 + \tilde{C} \Delta t)  \| \mathcal{I}_1^{n-1} - \mathcal{Q}_{1,R}^{n-1} \|+\tilde{C}\Delta t^2.
\end{align*}
The same argument as above provides the bound for the error in the second characteristic variable:
\begin{align*}
\| W_2^n - V_2^n \| &\leq (1+ \Delta t \tilde{C})\| W_2^{n-1} - V_2^{n-1} \| + \Delta t \tilde{C} \|W_1^{n-1} - V_1^{n-1}\| + \|\Pi_{h} V_2^{n} - V_2^{n} \| \\
& \qquad + \tilde{C} (1 + \tilde{C} \Delta t)  \| \mathcal{I}_2^{n-1} - \mathcal{Q}_{2,R}^{n-1} \|+\tilde{C}\Delta t^2.
\end{align*}

Summing the two above inequalities, and possibly increasing $\tilde{C}$, one obtains:
\begin{align*}
\|W_1^n - V_1^n\| + \| W_2^n - V_2^n \| &\leq (1 + \tilde{C} \Delta t )\Big\{\|W_1^{n-1} - V_1^{n-1}\| + \| W_2^{n-1} - V_2^{n-1} \|\Big\} \\
& \quad + \|\Pi_h V_1^n - V_1^n\| +  \|\Pi_{h} V_2^{n} - V_2^{n} \| \\ 
& \quad + \tilde{C}(1 + \tilde{C} \Delta t) \Big\{ \| \mathcal{I}_1^{n-1} - \mathcal{Q}_{1,R}^{n-1} \|+ \| \mathcal{I}_2^{n-1} - \mathcal{Q}_{2,R}^{n-1} \| \Big\}+ \tilde{C}\Delta t^2.
\end{align*}

We apply the same argument to successively bound the terms $\|W_1^{j}- V_1^{j}\| + \|W_2^j - V_2^j\|$ and conclude:
\begin{align*}
\| W_1^n - V_1^n \| + \|W_2^n - V_2^n\| &\leq \sum_{j=0}^n \exp( \tilde{C} \Delta t )^{n-j}\Big\{\| \Pi_{h} V_1^j - V_1^j \| + \| \Pi_{h} V_2^j - V_2^j \| \Big\} \\
& + \sum_{j=0}^{n-1}\exp(\tilde{C} \Delta t )^{n-j}\tilde{C} \Big\{\| \mathcal{I}_1^{j} - \mathcal{Q}_{1,R}^{j} \| + \| \mathcal{I}_2^{j} - \mathcal{Q}_{2,R}^{j} \| \Big\} + \sum_{j=0}^{n-1}  \tilde{C}\exp(\tilde{C} \Delta t)^{j}\Delta t^2  \\
&\leq \frac{T}{\Delta t} \exp(C T)\Big[\max_{i,j} \|\Pi_{h} V_i^j - V_i^j \| + \max_{i,j}\| \mathcal{I}_i^{j} - \mathcal{Q}_{i,R}^{j} \|+ O(\Delta t^2)\Big],
\end{align*}
where $C>0$ is a new constant, large enough such that we can take all the prefactors outside the parentheses.  The maximum is taken over $i = 1, 2$ and $j = 1, \ldots, n$. For the rectangle rule, one can show:
\begin{equation}
\max_j \|\mathcal{I}_i^j - \mathcal{Q}_{i,R}^j\| \leq C_{V} \frac{\Delta t^2}{2}
\end{equation}
where $C_{V}=C_{V}(V_{1},V_{2})$. For piecewise linear interpolation, we have:
\begin{equation}
\max_j \|\Pi_h V_i^j - V_i^j\| \leq \frac{h^2}{8}\|V_i''\|_T.
\end{equation}
With these bounds we obtain the result.
\end{proof}

\begin{myrem}
Practically we take $h$ proportional to $\Delta t$, so the error decreases linearly in both $\Delta t$ and $h$. Notice that neither the Stability Proposition \ref{Prop.Stability} nor the Convergence Theorem \ref{Thm.Convergence} are dependent on the choice for the constant of proportionality. In fact, in order to obtain convergence at a linear rate, it is only needed that $h / \Delta t$ is bounded above. In other words, our proposed method is unconditionally stable with no need to satisfy a CFL--type condition.
\end{myrem}


\section{Transmission Conditions at Branching Points} \label{Section:Trans}

The end goal of the one-dimensional blood flow models is to simulate hemodynamics in a network of one-dimensional vessels representing portions of the circulatory system. These vessels are connected at nodes or branching points where the flow is governed by conservation laws. Various models have been proposed to simulate the branching flows. We refer to \cite[Section 3.1]{FLQ03}. We simply impose conservation of mass and continuity of the total pressure at each interior node of the network.

In general, let $J$ be the number of incoming and outgoing vessels at a given node, and $(A_{j},u_{j})$ the cross-sectional area and flow velocity respectively for each vessel indexed by $j=1,...,J$. Without loss of generality, we assume the 1d coordinates on each vessel to be such that blood flows out of the node for positive values of the velocities $u_{j}$.  Conservation of mass requires that
\begin{align} 
\label{eqn:conservmass}
\sum_{j=1}^{J} A_{j} u_{j} = 0,
\end{align}
whereas continuity of total pressure is enforced by the following equations
\begin{align} 
\label{eqn:contpress}
\frac{1}{2} u_{1}^2 + p_{1} / \rho = \frac{1}{2} u_{j}^2 + p_{j} / \rho, \quad j=2,...,J, 
\end{align}
where $p_{j} = p_{j}(A_{j})$ is defined by (\ref{eq:state}). The goal is to translate these physical conservation laws into the transmission of characteristic variables at the connecting node. Recall that on each branch we have a pair of characteristics, one traveling out of the node and another into the node. We denote them as $W_{+,j}^{n}$ and $W_{-,j}^n$, respectively, where $n$ is the time step to be computed. Since $W_{-,j}^{n}$ travels into the node, then it can be determined explicitly from the information at the $n-1$ time level using the Algorithm \ref{alg:NMC}. Hence, by plugging (\ref{eq:Ch_vs_Phys}) into (\ref{eqn:conservmass})--(\ref{eqn:contpress}), we obtain a nonlinear system of $J$ algebraic equations for the unknowns $W_{+,1}^{n},W_{+,2}^{n},...,W_{+,J}^{n}$ which we solve with Newton's method. This approach constitutes our numerical transmission conditions for the characteristic variables at each node of a network. In our numerical implementation of these transmission conditions, we use $W_{+,1}^{n-1},W_{+,2}^{n-1},...,W_{+,J}^{n-1}$ as the initial guess for Newton's method, and we stop the iterative process when the relative difference between two consecutive iterations falls below a certain tolerance. In the simulations described in the next section, we select the tolerance to be $10^{-8}$ which is much smaller than the expected error introduced by the discretization of the spatial and temporal domains.


\section{Numerical Experiments} \label{Section:NumExp}

\subsection{Convergence rate and unconditional stability} \label{Section:Conv}

We compute the convergence rate of our method by comparing our numerical solution to the exact solution
\begin{align*}
A(x,t) &= \left(1 + t \exp(-10 t) \sin\frac{\pi x}{L}\right)^2\\
u(x,t) &= 0
\end{align*}
with boundary conditions $A = A_0 = 1 \text{ cm$^{2}$}$ on the inlet and outlet. The spatial variable $x \in [0,L]$ for $L= 20 \text{ cm}$. The time variable $t \in [0,T]$ where $T = 1 \text{ sec}$. The characteristic variables $V_1$ and $V_2$ are then derived from (\ref{eq:v1v2_1}) -- (\ref{eq:v1v2_2}). Recall that $V_{1}$ propagates to the right and $V_{2}$ to the left, so we impose a boundary condition for $V_{1}$ at $x=0$ and for $V_{2}$ at $x = L$. Since $A=A_{0}$ at the boundary points $x=0$ and $x=L$, then $c=c_{0}$ at those two points, and the appropriate boundary condition for the numerical variables are obtained from  (\ref{eq:Ch_vs_Phys}) as follows,
\begin{align*}
W_{1}^{n}|_{x=0} = W_{2}^{n}|_{x=0} \quad \text{and} \quad W_{2}^{n}|_{x=L} = W_{1}^{n}|_{x=L} \quad \text{for all $n=1,...,N$},
\end{align*}
where $W_{2}^{n}|_{x=0}$ and $W_{1}^{n}|_{x=L}$ are explicitly given from the previous time step using the Algorithm \ref{alg:NMC}.

Following the test case presented in \cite{MN08}, the parameters are chosen as $\beta = 229674 \text{ dyne}/\text{cm}^3$ and $\nu = 0$. Using the standard approach, we derive the source terms for this exact solution and then compute a numerical approximation with NMC.

To highlight the perfomance of the method beyond the traditional CFL limitation, let us consider the following constant
\begin{equation}
K_{\rm CFL} := \frac{c_{0}\Delta t}{h}, \qquad \text{where} \quad c_{0}^{2} = \frac{\beta}{2 \rho} A_0^{1/2}.
\end{equation}  
Here $c_{0}$ approximates the speed of pressure waves. Explicit methods require $K_{\rm CFL}$ to be bounded (typically less than 1) for stability, but our method requires no such restriction.  In this light, we set $K_{\rm CFL} = 2^n$ to investigate the convergence behavior of the method as $n$ increases. Table \ref{tab:conv} displays relative error in the supremum norm (over space and time) and convergence rate for different values of $K_{\rm CFL}$, and $h = L / 2^{3+m}$ and
$\Delta t = K_{\rm CFL} h / c_{0}$ for $m = 1, \ldots, 6$.

\begin{table}[H]
\centering
\scalebox{0.9}{
\begin{tabular}{l|c c c c c c c}
\hline
\hline
& & & &\textbf{Rel Error} & & & \\
$m$ & $K_{\rm CFL}=1/4$ & $K_{\rm CFL}=1/2$ & $K_{\rm CFL}=1$ & $K_{\rm CFL}=2$ & $K_{\rm CFL}=4$ & $K_{\rm CFL}=8$ & $K_{\rm CFL}=16$ \\
\hline
$1$ & $	1.77 \times 10^{-3}	$ & $	1.25\times 10^{-3}	$ & $	2.73\times 10^{-3}	$ & $	5.53\times 10^{-3}	$ & $	1.10\times 10^{-2}	$ & $	2.07\times 10^{-2}	$ & $	7.22\times 10^{-2}	$ \\
$2$ & $	9.80\times 10^{-4}	$ & $	6.32\times 10^{-4}	$ & $	1.38\times 10^{-3}	$ & $	2.77\times 10^{-3}	$ & $	5.57\times 10^{-3}	$ & $	1.10\times 10^{-2}	$ & $	2.08\times 10^{-2}	$ \\
$3$ & $	5.18\times 10^{-4}	$ & $	3.19\times 10^{-4}	$ & $	6.92\times 10^{-4}	$ & $	1.39\times 10^{-3}	$ & $	2.78\times 10^{-3}	$ & $	5.59\times 10^{-3}	$ & $	1.11\times 10^{-2}	$ \\
$4$ & $	2.66\times 10^{-4}	$ & $	1.60\times 10^{-4}	$ & $	3.47\times 10^{-4}	$ & $	6.95\times 10^{-4}	$ & $	1.39\times 10^{-3}	$ & $	2.78\times 10^{-3}	$ & $	5.59\times 10^{-3}	$ \\
$5$ & $	1.35\times 10^{-4}	$ & $	8.03\times 10^{-5}	$ & $	1.74\times 10^{-4}	$ & $	3.48\times 10^{-4}	$ & $	6.96\times 10^{-4}	$ & $	1.39\times 10^{-3}	$ & $	2.78\times 10^{-3}	$ \\
$6$ & $	6.80\times 10^{-5}	$ & $	4.02\times 10^{-5}	$ & $	8.69\times 10^{-5}	$ & $	1.74\times 10^{-4}	$ & $	3.48\times 10^{-4}	$ & $	6.96\times 10^{-4}	$ & $	1.39\times 10^{-3}	$ \\
\hline
\hline
& & & & \textbf{Conv Rate} & & & \\
$2$ & $0.85$ & $0.98$ & $0.99$ & $1.00$ & $0.98$ & $0.91$ & $1.80$ \\
$3$ & $0.92$ & $0.99$ & $0.99$ & $1.00$ & $1.00$ & $0.98$ & $0.91$ \\
$4$ & $0.96$ & $0.99$ & $1.00$ & $1.00$ & $1.00$ & $1.01$ & $0.98$ \\
$5$ & $0.98$ & $1.00$ & $1.00$ & $1.00$ & $1.00$ & $1.00$ & $1.01$ \\
$6$ & $0.99$ & $1.00$ & $1.00$ & $1.00$ & $1.00$ & $1.00$ & $1.00$ \\
\hline
\end{tabular}}
\caption{Relative errors in the supremum norm (over space and time) for the NMC. The asymptotic linear rate of convergence proven in Theorem \ref{Thm.Convergence} is observed in these numerical experiments for increasing values of the CFL number $K_{\rm CFL}$. These experiments confirm the unconditional stability of the NMC. \label{tab:conv}}
\end{table}

\subsection{Single uniform vessel} 

In this section, we compare the numerical method of characteristics applied to (\ref{eq:orig1})--(\ref{eq:orig2}) for approximating ($V_1$, $V_2$) to a discontinuous Galerkin (dG) discretization applied to (\ref{eq:model})--(\ref{eq:state}) for approximating ($A$, $u$) (as described by Sherwin et al. \cite{SFPF03}).  The computational domain is a single vessel of length $20$ cm.  The vessel parameters are again derived from the test case presented in \cite{MN08}; $A_0 = 1\text{ cm$^{2}$}$, $\beta = 229674 \text{ dyne}/\text{cm}^3$ and density $\rho = 1.06 \text{ g/cm}^3$. Further, we set the viscosity $\nu = 0$ so that we can attribute any possible diffusion to the numerical method itself. An initial Gaussian pressured pulse in time is prescribed at the left inlet of the vessel with functional form
\begin{equation}
p(t) = \alpha \exp \Big((t-\xi)^2/2\sigma^2\Big). \label{eqn:GaussianProfile}
\end{equation}
The parameters $\alpha = 10^2$ or $10^3$ $\text{ dyne/cm}^2$, $\xi = 0.015 \text{ s}$, $\sigma = 0.003 \text{ s}$ remain the same for each numerical experiment in this section. The procedure for prescribing incoming boundary conditions for the dG method is described in \cite{SFPF03}. For the NMC, from the pressure profile (\ref{eqn:GaussianProfile}), one derives the prescribed area $A$ at the inlet from the state equation (\ref{eq:state}) (or equivalently the local wave speed $c$ from (\ref{eqn:speed})). From (\ref{eq:Ch_vs_Phys}) then we obtain the inlet boundary condition $W_{1}^{n} = W_{2}^{n} + 8 \left( c(t_{n}) - c_{0} \right)$ where $W_{2}^{n}$ is explicitly obtained from the information at the $n-1$ time level using the Algorithm \ref{alg:NMC}. The outlet boundary condition is of absorbing type, that is, the waves are allowed to leave the domain without reflection by setting $W_{2}^{n} = 0$ at $x=L$ for all $n=1,..., N$.

As a metric for comparing the approximate solutions obtained from NMC and dG, define the vectors ${\bf p}_{\rm dG}$ and ${\bf p}_{\rm NMC}$ as the pressures computed from each method with each component corresponding to a pressure value at a point in the NMC grid $G_h$. Then the relative difference is given by $
\| {\bf p}_{\rm dG} - {\bf p}_{\rm NMC}\|_2 / \|{\bf p}_{\rm dG}\|_2$,
where $\|\cdot\|_2$ is the vector two-norm. Figure \ref{fig:1dvessel} displays the approximate solutions to both methods for $\alpha = 10^2$ (no shock) and $\alpha = 10^3$ (shock) respectively.  Visually, they appear to agree well, modulo some small diffusion in the NMC solution. Table \ref{tab:compare_error} displays the relative difference between the dG and NMC solutions at each of the times $t = 0.03$, $0.045$, and $0.06$, and confirms the agreement of the solutions.  The two methods agree less well in capturing the shock, but we note that shock formation is not physiological for normal blood flow.

\begin{table}[H]
\centering
\scalebox{1.0}{
\begin{tabular}{c | c c c}
\hline
\hline
& & $\| {\bf p}_{\rm dG} - {\bf p}_{\rm NMC}\|_2 / \|{\bf p}_{\rm dG}\|_2$ & \\
$\alpha$ &$t = 0.03$ &$t = 0.045$ &$t = 0.06$  \\
\hline
$10^2$& $2.78\times 10^{-3}$&$3.71 \times 10^{-3}$ &$4.95\times10^{-3}$ \\
$10^3$&$1.01\times10^{-2}$ &$3.03\times10^{-2}$ &$8.53\times 10^{-2}$  \\ 
\hline
\end{tabular}}
\caption{Relative difference in dG and NMC solutions for simulations within a single vessel. \label{tab:compare_error}}
\end{table}

\begin{figure}[H]
\begin{center}
\includegraphics[scale=0.45, trim=22 175 70 300]{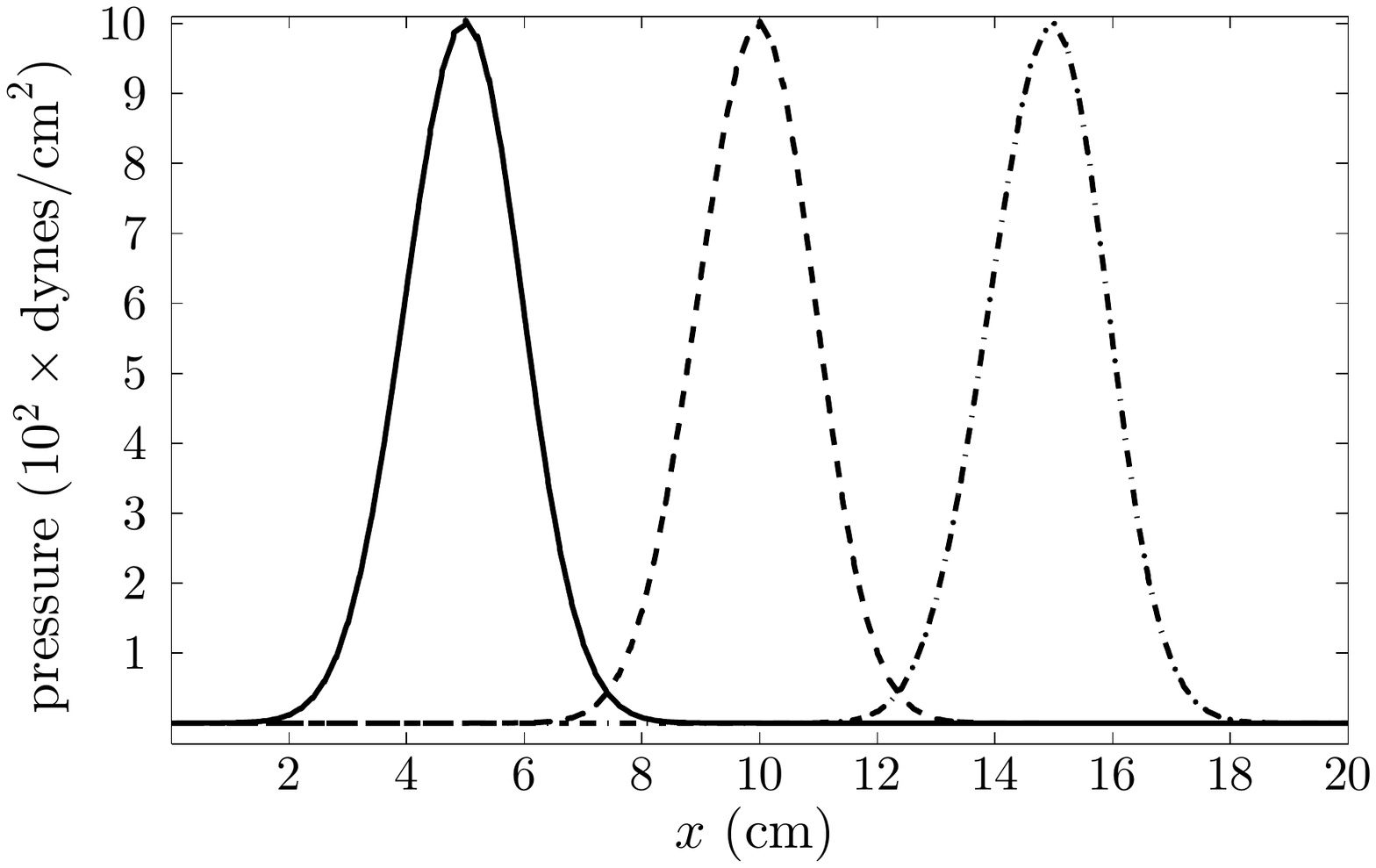}
\put(-130,158){\small dG for $(A,u)$}
\includegraphics[scale=0.45, trim=22 175 70 300]{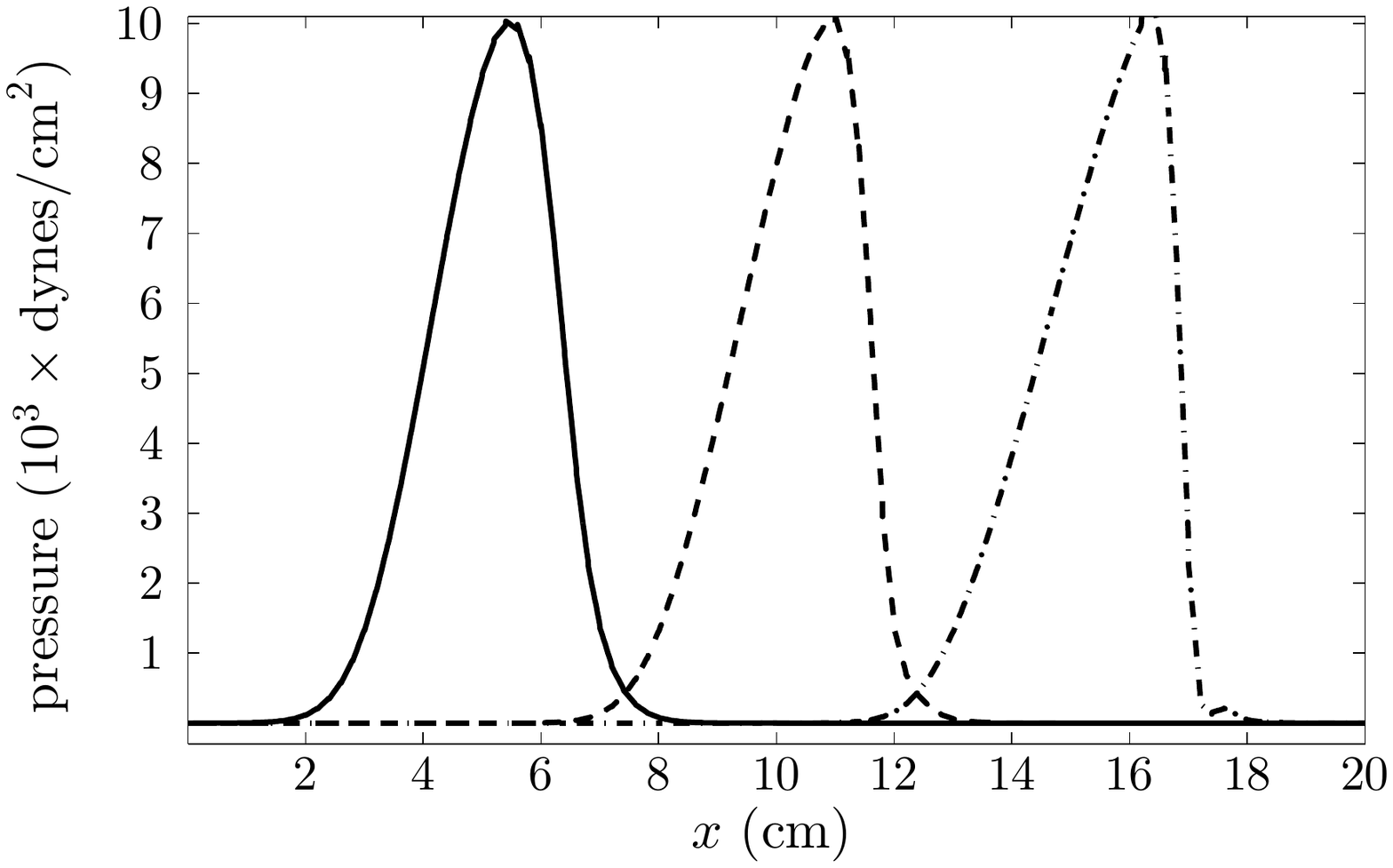}
\put(-130,158){\small dG for $(A,u)$}\\
\vspace{0.1cm}
\includegraphics[scale=0.45, trim=22 215 70 300]{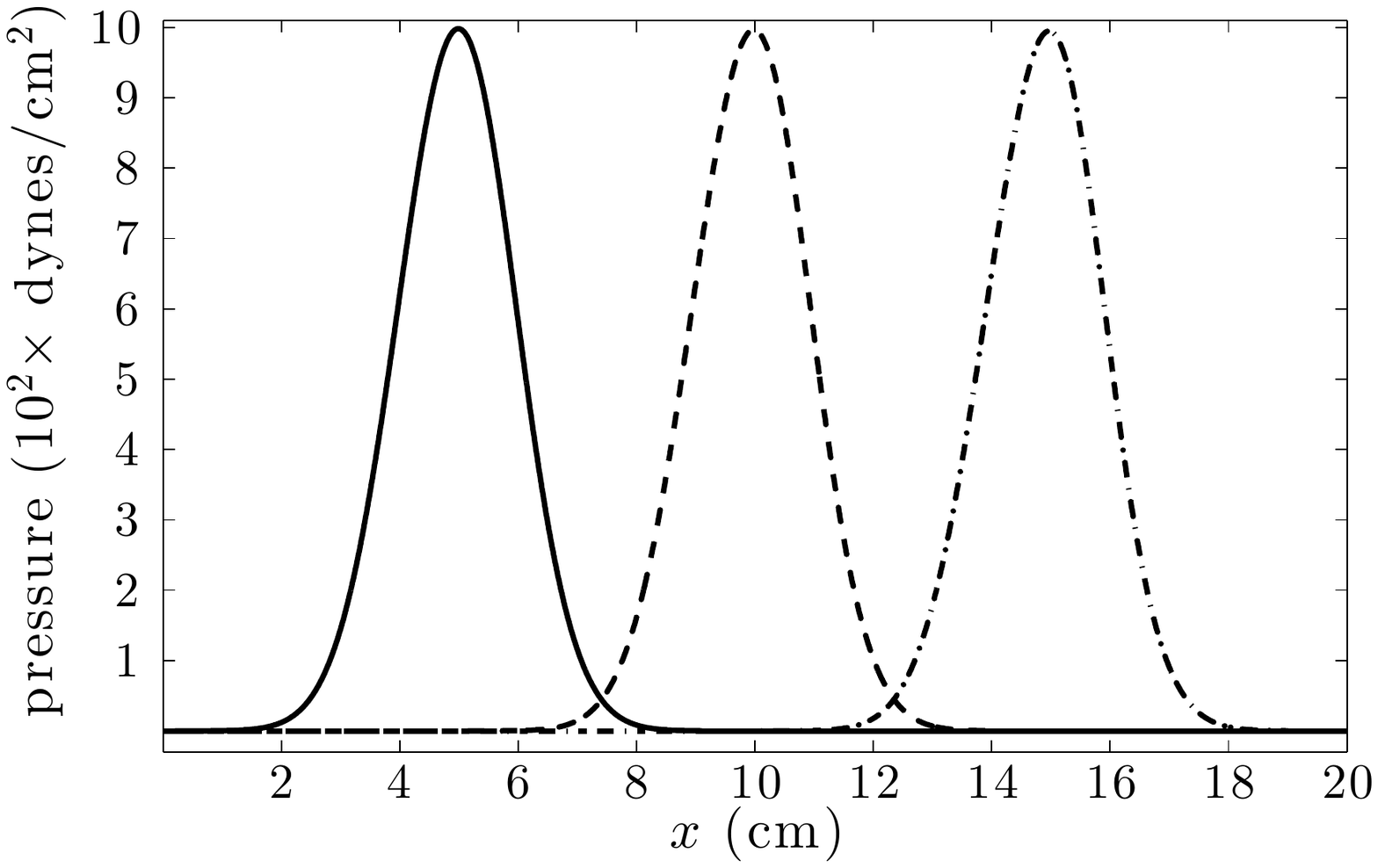}
\put(-136,140){\small NMC for $(V_1,V_2)$}  
\includegraphics[scale=0.45, trim=22 215 70 300]{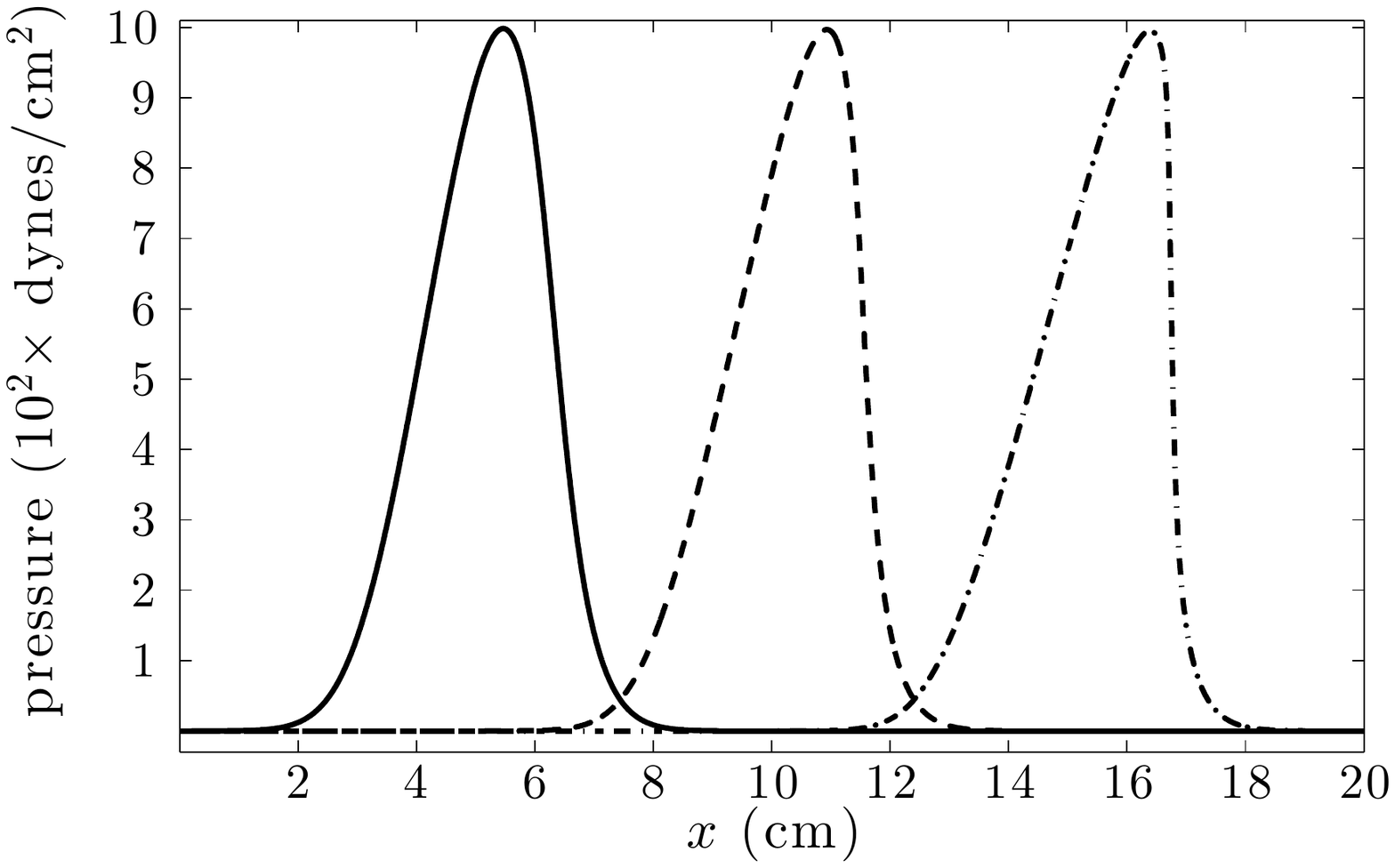}
\put(-136,140){\small NMC for $(V_1,V_2)$}  
\end{center}
\caption{
Propagation of Gaussian pressure pulse in a single uniform vessel, simulated with discontinuous Galerkin (top) and numerical method of characteristics (bottom).  The dG implementation uses $100$ elements with piecewise linear polynomials and $\Delta t = 1 \times 10^{-5}$, while the NMC implementation uses a spatial discretization of $1 \times 10^{-2}$ and $\Delta t = 1 \times 10^{-4}$.  Time increases from the left to right with snapshots taken every 0.03 seconds, i.e. for the solid line, $t = 0.03$, for the dashed line, $t = 0.045$, and for the dotted-dashed line, $t = 0.06$. As expected, the NMC method exhibits  some small numerical dissipation, but agrees very well the with dG simulation.  For the figures on the right, the amplitude of the wave is an order of magnitude larger than in the figures on the left, leading to rapid shock formation within the computational domain.}
\label{fig:1dvessel}
\end{figure}

Lastly, Figure \ref{fig:timing} displays timing results for \textsc{Matlab} implementations of each method applied to the simulation of a pressure pulse in a single vessel.  For both cases, $\Delta t = 1\times 10^{-6}$ and the degrees of freedom ($DOF$) for each method are defined as follows,
\begin{align*}
& DOF_{\rm dG} = \Big\{\text{number of elements}\Big\} \times \Big\{\text{polynomial degree $+$ 1}\Big\}, \\
& DOF_{\rm NMC} = \text{number of points in } G_h.
\end{align*}
We integrate the solution for $20$ timesteps (the final time $T = 2 \times 10^{-5}$ sec.) on a laptop with a $2.5$ GHz Intel Core i5-2520M processor.  The value displayed in Figure \ref{fig:timing} is wall clock time, averaged over 25 realizations, normalized by $T$, and then divided by $DOF$. As expected, both methods are asymptotically linear in $DOF$, with NMC several of orders of magnitude faster than dG.

\begin{figure}[h!]
\begin{center}
\includegraphics[scale=0.60, trim=50 215 50 450]{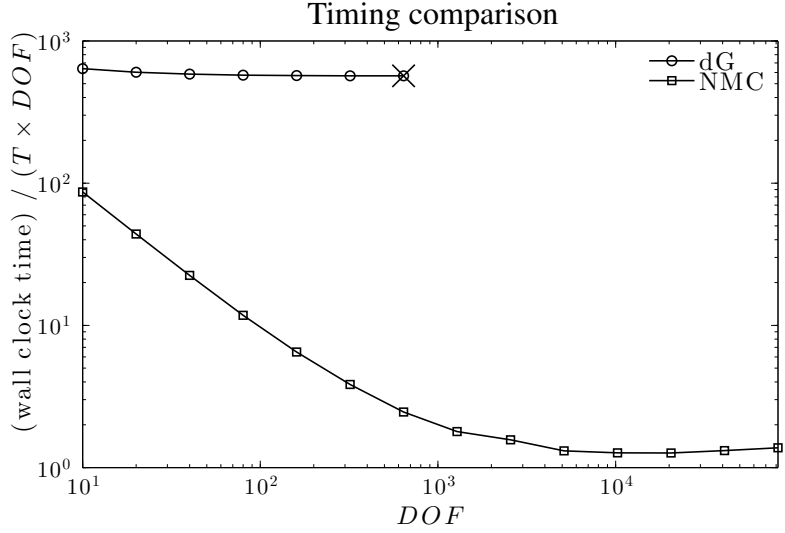}
\put(-200,185){Timing comparison}
\end{center}
\caption{A comparison of the computational time for NMC and dG methods applied to the simulation of a pressure pulse in a single vessel.  Both methods asymptotically scale linearly in $DOF$, with NMC substaintially faster than dG.  The `$\times$' on the dG curve indicates that for this spatial discretization and beyond, we cannot expect the method to be stable for the chosen timestep.}
\label{fig:timing}
\end{figure}

\subsection{Vessel networks}

In this section we demonstrate the utility of the numerical method of characteristics in simulating flow in a network of vessels, each modeled by (\ref{eq:model})--(\ref{eq:state}).

First we set up a small network ($5$ branches and $2$ interior nodes) which represents the large arteries in the left arm.  The parameters are taken from \cite{SFPF03}. We use this small network to compare the NMC with the dG method in the presence of branching points at which we enforce the transmission conditions of Section \ref{Section:Trans}. To validate this proposed transmission conditions for the NMC, we compare the results obtained from the dG method and the NMC applied to this five vessel network. The pressure at the input node and at one of the terminal nodes is displayed in Figure \ref{fig:smalltree}, along with the relative difference between the two numerical solutions. From this figure we observe that both methods compare well since the relative difference is below the $2 \%$ mark. We take into account the blood viscosity whose value is set to $\nu = 3.3 \times 10^{-2} \text{ cm$^2$ / s}$. The spatial and temporal step sizes for the NMC are $h = 1$ cm and $\Delta t = 2.5 \times 10^{-3}$ s, respectively. For the dG method, $h = 1$ cm and $\Delta t = 10^{-4}$ s, and we use piecewise linear polynomials.

\begin{figure}[H]
\begin{center}
\includegraphics[height = 0.4 \textheight]{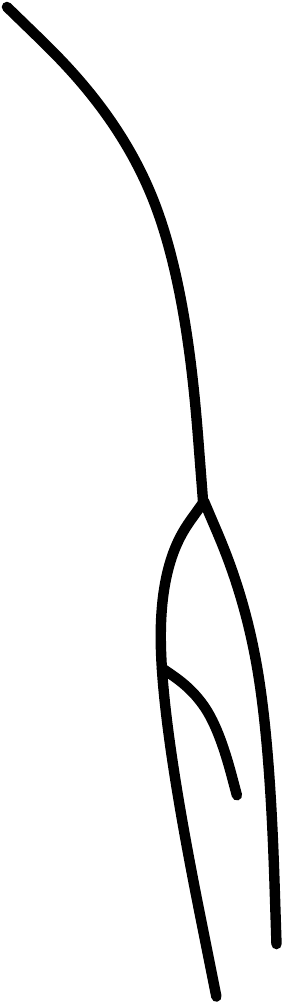}
\hspace{1cm}\includegraphics[height = 0.4 \textheight, trim=50 50 100 50]{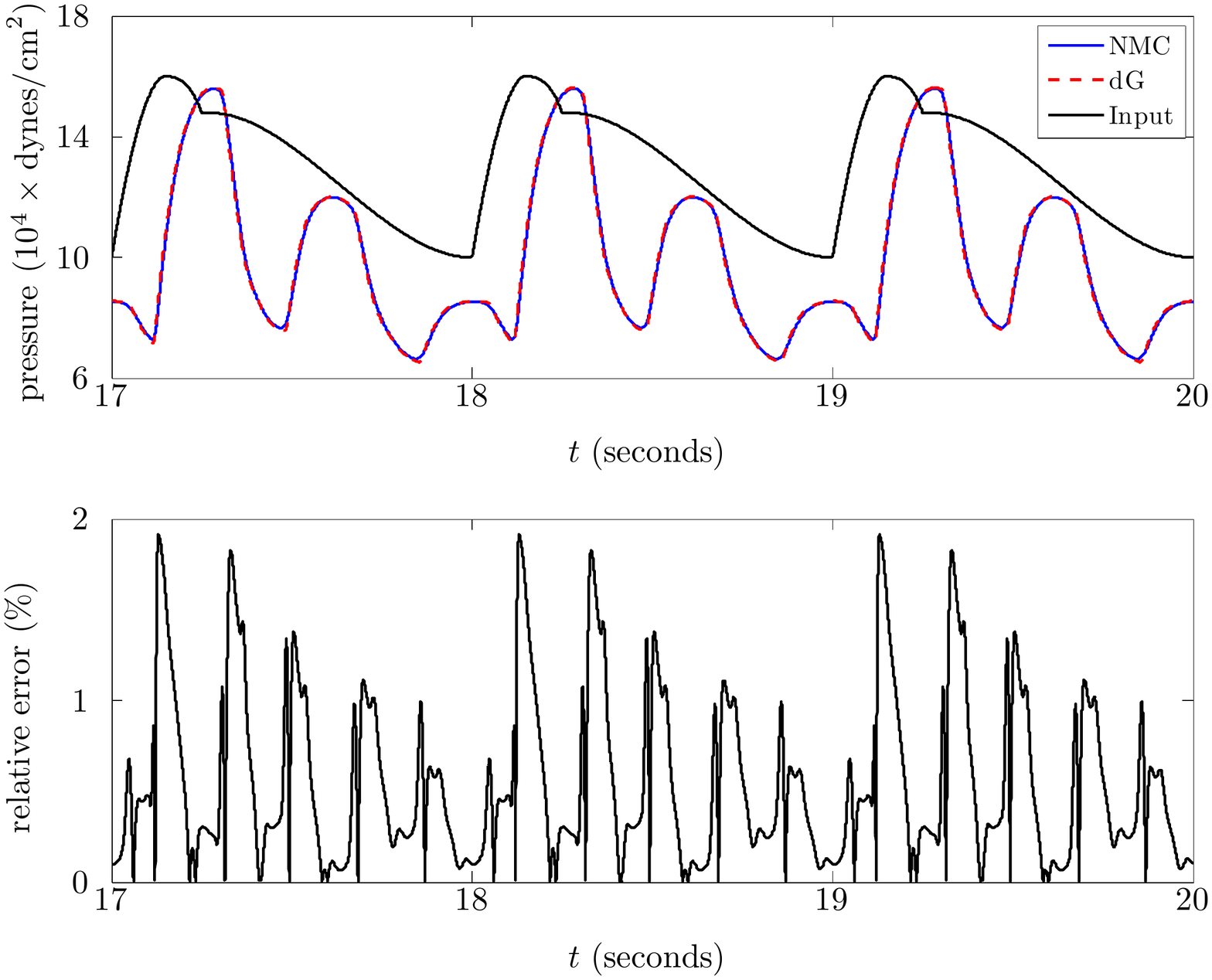}
\end{center}
\caption{Left: A sketch of a small arterial network representing the left arm. The geometric and elastic properties of the vessels are provided in \cite{SFPF03}. Top right: Pressure waveforms at the top node of the network and terminal node of the left radial artery. There is no resistance imposed at the terminal ends. Bottom right: Percent difference in pressure between NMC and dG, relative to the norm of the dG solution, $|p_{\rm dG} - p_{\rm NMC}| / \| p_{\rm dG} \| \times 100$.}
\label{fig:smalltree}
\end{figure}

As a second example, we set up the arterial network from \cite{Myn11} which contains the 64 largest arteries in the human body (we exclude coronary arteries). For sake of simplicity, we do not incorporate the influence of organs, capillary beds or the venous network. There is no resistance imposed at the terminal ends of this arterial model where the pressure waves are allowed to leave the terminal vessels without reflection. We do take into account the blood viscosity by retaining the zeroth order (dissipative) term of the governing system (\ref{eq:1dchar}), where we set $\nu = 3.3 \times 10^{-2} \text{ cm$^2$ / s}$. The length and radius of each arterial segment is obtained from \cite{Myn11}. The elastic coefficient $\beta$ of each segment is given by the following empirical formula,
\begin{equation*}
\beta = \frac{\delta}{r} \frac{E}{(1-\sigma^2) r}, \qquad \text{and} \quad E = \frac{r}{\delta} \left( 6 \times 10^6 e^{- 9 r / \text{ cm} } + 33.7 \times 10^4  \right) \, \text{ dyne } / \text{ cm}^{2}.
\end{equation*}
where $\delta/r = 0.1$ is the ratio of wall thickness $\delta$ to unperturbed cross-sectional radius $r$. The Poisson's ratio is $\sigma = 0.5$, and $E$ is the Young's modulus of elasticity.

Figure \ref{fig:arterialtree} displays the input pressure profile at the Aortic root and the observed pressure at the left Radial artery.
The simulations were carried out with quasi-uniform spatial discretizations parametrized by $h$ and time step $\Delta t_{\rm NMC}$. The parameters $h$ and $\Delta t_{\rm NMC}$ were refined proportionally, but in all three cases $\Delta t_{\rm NMC}$ is sufficient small to appropriately resolve the pressure variations within one cardiac cycle.  The three solid lines in Figure \ref{fig:arterialtree} display the convergence behavior as the spatial and temporal steps are refined.

From the given geometry and elastic properties of this arterial tree, we obtain a pressure wave speed $c_{0}$ varying within the following range $460 - 1300 \text{ cm/sec}$. As a result, for the chosen $h$ and $\Delta t_{\rm NMC}$, we have a maximum CFL number $K_{\rm CFL} \approx 6.5$.
On the other hand, the time step needed to satisfy stability for a piecewise linear explicit dG scheme is known to be
\begin{equation*}
\Delta t_{\rm dG} < \frac{h}{3 \max{c_{0}}} \approx 1.66 \times 10^{-4}.
\end{equation*} 
This implies that $\Delta t_{\rm NMC} = 2 \times 10^{-3}$ sec (the intermediate refinement in Figure \ref{fig:arterialtree}) is about $12$ times larger than $\Delta t_{\rm dG}$. The spatial discretization $h = 1 \text{ cm}$ leads to about $900$ degrees of freedom (DOF) for the NMC method applied to the entire arterial tree. If we consider both the gain in computational speed per DOF (displayed in Figure \ref{fig:timing}) and the larger time step allowed by the unconditional stability of the NMC, then we conclude that the NMC is at least $3$ orders of magnitude more efficient than the dG method for these physiological parameters.

\begin{figure}[H]
\begin{center}
\includegraphics[height = 0.35 \textheight]{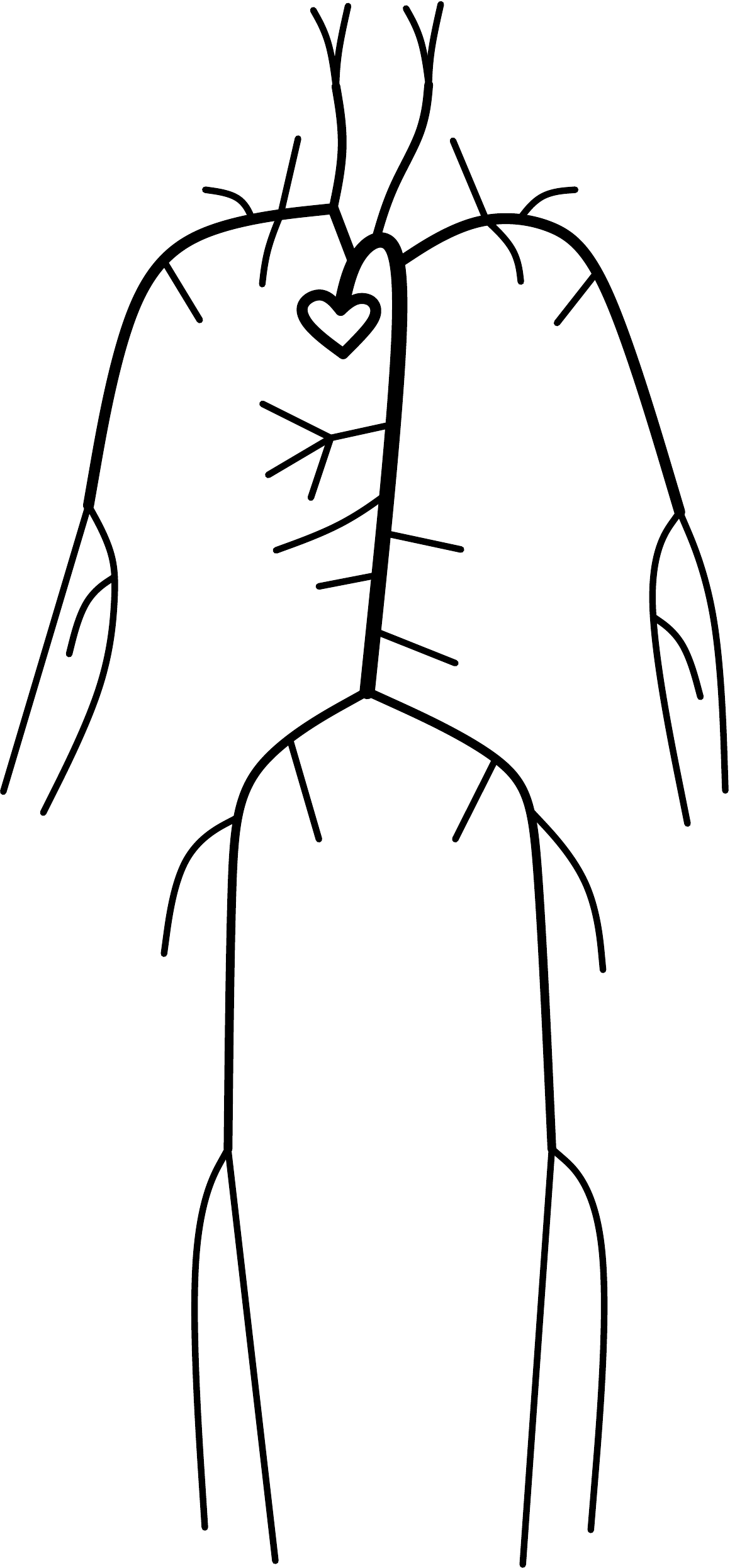}
\hspace{1cm}\includegraphics[height = 0.4 \textheight, trim=80 50 100 100]{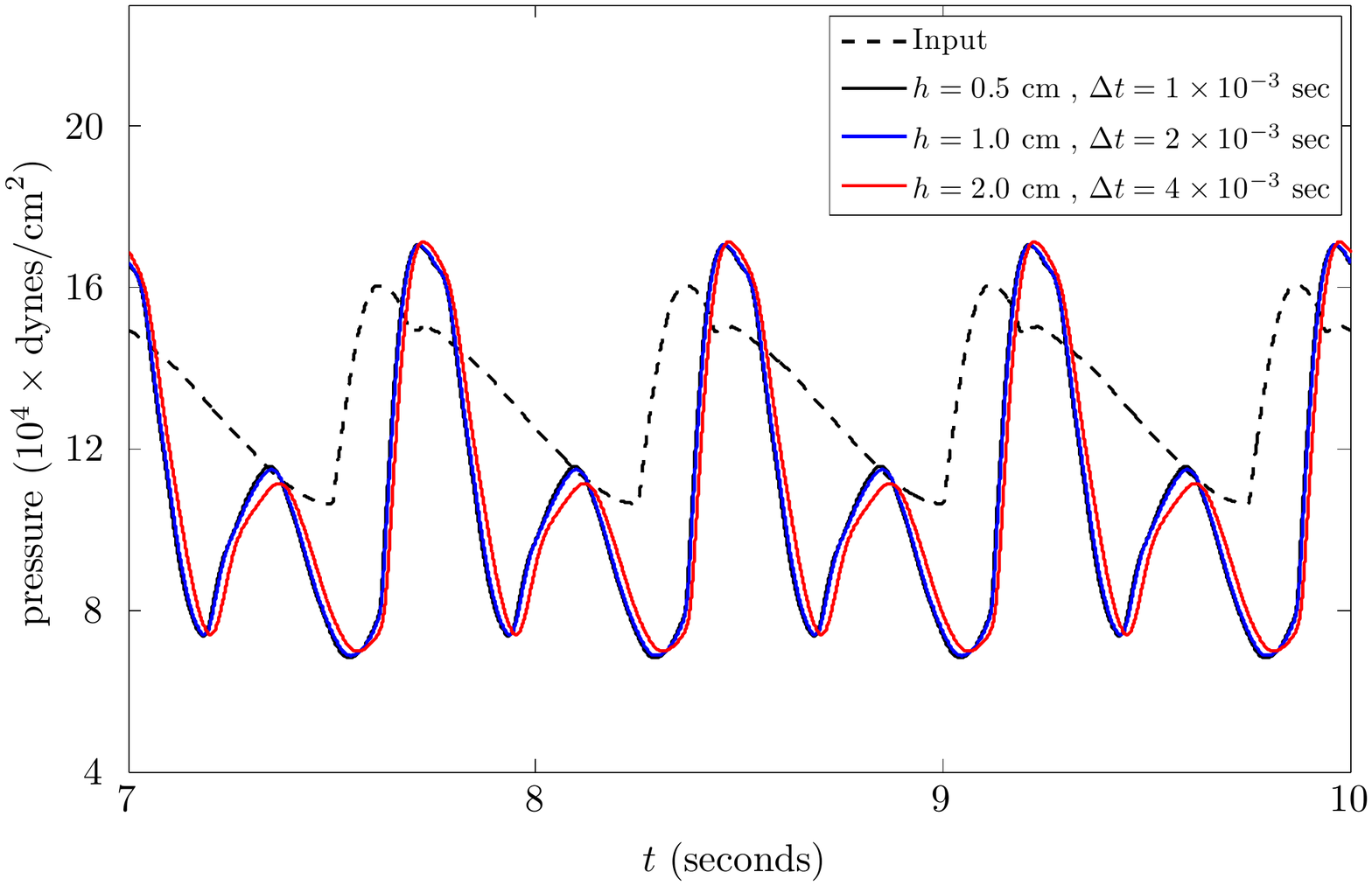}
\end{center}
\caption{Left: A sketch of the systemic arterial network containing 64 segments and 29 interior nodes. The geometric and elastic properties of the vessels are provided in Table F.1 of \cite{Myn11}. Right: Pressure waveform (at the aorta (dashed) and left radial (solid) arteries) obtained from the simulation based on the NMC. There is no resistance imposed at the terminal ends. The three solid lines display the convergence behavior as the spatial and temporal steps are refined.}
\label{fig:arterialtree}
\end{figure}


\section{Conclusion}

In this work, we focused on the numerical approximation of solutions to a nonlinear, strictly hyperbolic system modeling one-dimensional blood flow.  Typical physiological parameters lead to large pressure wave speeds and hence to a restrictive CFL condition for methods using explicit time stepping for the primitive governing equations.  This stringent condition is magnified for computationally intensive methods arising from weak formulations, in simulations of networks of vessels, and for simulations required over multiple cardiac cycles.

To mitigate these challenges, we presented a numerical method of characteristics approach applied to this system.  Unconditional stability and convergence of the method was proven.  The unconditional stability allows for more rapid simulations beyond the traditional CFL limitation.

To benchmark and test our method, we computed errors and convergence rates from a specified exact solution. Further, solution quality for a propagating Gaussian pressure pulse was compared to an approximation from a discontinuous Galerkin implementation.  As expected, numerical diffusion occurs in our method for coarse spatial discretizations, but a marginally more refined discretization yields much better results. Lastly, we applied the method to a network of vessels. From the timing results for the dG and NMC implementations, and due to the larger time step allowed for NMC, we conclude that NMC is at least 1000 times more efficient than dG.

Future work will entail clinical applications of vessel network simulations including the influence of organs, capillary beds, and the venous network. These full cardiovascular models, simulated with the numerical method of characteristics, will allow researchers and clinicians to investigate challenging physiological questions from a computational modeling perspective.  Furthermore, the efficiency of our approach allows for simulations over a large number of heart cycles on modestly sized computers. In turn, this opens a door for a much more computationally tractable approach for modeling these phenomena.


\section{Acknowledgments}
This work was funded in part by NSF grant NSF-DMS 1312391 and by a training fellowship from the Keck Center of the Gulf Coast Consortia, on the Training Program in Biomedical Informatics, National Library of Medicine (NLM) T15LM007093.






\bibliographystyle{elsarticle-num}
\bibliography{nmc}

\end{document}